%% file: main.tex
\documentclass[11pt]{article}

\def\final{1}

\usepackage{lmodern}
\usepackage{amssymb,amsmath,amsthm}
\usepackage[T1]{fontenc}

\usepackage[margin=1in]{geometry}
\usepackage{algorithm}

\usepackage{paralist}

\usepackage{verbatim}
\usepackage{accents}

\usepackage{xcolor}
\definecolor{linkcol}{rgb}{0,0,0.35}
\definecolor{todocol}{rgb}{0.6,0,0}
\definecolor{citecol}{rgb}{0.1,0.35,0}

\usepackage{hyperref}
\usepackage{cleveref}
\usepackage{tikz}
\usetikzlibrary{arrows.meta}
\usetikzlibrary{positioning,calc}

\hypersetup{
            pdftitle={Algorithms for flows over time with scheduling costs},
            pdfborder={0 0 0},
            breaklinks=true}
\newcommand{\ipcoskip}[1]{#1}
\newcommand{\ifipco}[2]{#2}

\ifdefined\final
\newcommand{\nnote}[1]{\ignorespaces}
\newcommand{\dnote}[1]{\ignorespaces}
\newcommand{\todo}[1]{\ignorespaces}
\else
\newcommand{\nnote}[1]{{\em\color{blue!50!black}Neil: #1}}
\newcommand{\dnote}[1]{{\em\color{purple}Dario: #1}}
\newcommand{\todo}[1]{{\em\color{red!50!black}TODO: #1}}
\fi

\definecolor{mydarkgreen}{RGB}{0,100,0}

\theoremstyle{theorem}
\newtheorem{theorem}{Theorem}
\newtheorem{lemma}[theorem]{Lemma}
\newtheorem{claim}[theorem]{Claim}

\newtheorem{corollary}[theorem]{Corollary}

\theoremstyle{definition}

\newtheorem{problem}{Problem}

\usepackage{algpseudocode}

\newenvironment{proofof}[1]{\begin{proof}[Proof of #1]}{\end{proof}}

\newcommand{\ind}{\mathbf{1}}

\newcommand{\lpint}{\int}
\newcommand{\lpsum}{\sum}

\newcommand{\rev}[1]{\overleftsmallarrow{#1}}
\newcommand{\bid}[1]{\overleftrightsmallarrow{#1}}

\newcommand{\R}{\mathbb{R}}
\newcommand{\Rplus}{\R_+}
\DeclareMathOperator{\cost}{cost}
\newcommand{\intR}{\int_{\R}}
\newcommand{\zg}{z}

\makeatletter
\newcommand{\overleftrightsmallarrow}{\mathpalette{\overarrowsmall@\leftrightarrowfill@}}
\newcommand{\overrightsmallarrow}{\mathpalette{\overarrowsmall@\rightarrowfill@}}
\newcommand{\overleftsmallarrow}{\mathpalette{\overarrowsmall@\leftarrowfill@}}
\newcommand{\overarrowsmall@}[3]{%
  \vbox{%
    \ialign{%
      ##\crcr
      #1{\smaller@style{#2}}\crcr
      \noalign{\nointerlineskip}%
      $\m@th\hfil#2#3\hfil$\crcr
    }%
  }%
}
\def\smaller@style#1{%
  \ifx#1\displaystyle\scriptstyle\else
    \ifx#1\textstyle\scriptstyle\else
      \scriptscriptstyle
    \fi
  \fi
}
\makeatother

\newcommand{\keywords}[1]{}

\title{Algorithms for flows over time with scheduling costs\thanks{%
Partially supported by NWO TOP grant 614.001.510 and NWO Vidi grant 016.Vidi.189.087.}}

\author{Dario Frascaria\footnote{Department of Econometrics \& Operations Research, Vrije Universiteit Amsterdam, Amsterdam, The Netherlands. \texttt{d.frascaria@vu.nl}.}~  and 
    Neil Olver\footnote{Department of Mathematics, London School of Economics and Political Science, London, UK, 
and CWI, Amsterdam, The Netherlands. \texttt{N.Olver@lse.ac.uk}.}}

\date{\today}

\begin{document}
\maketitle

\input{abstract}

\input{introduction}
\input{preliminaries}

\input{algorithm}

\input{duality}

\input{analysisNew}

\input{tolls}

\input{generalCase}

\bibliographystyle{abbrv}
\bibliography{biblio}

\end{document}

%% file: abstract.tex
\begin{abstract}
Flows over time have received substantial attention from both an optimization and (more recently) a game-theoretic perspective.
In this model, each arc has an associated delay for traversing the arc, and a bound on the rate of flow entering the arc; flows are time-varying.
We consider a setting which is very standard within the transportation economic literature, but has received little attention from an algorithmic perspective.
The flow consists of users who are able to choose their route but also their departure time, and who desire to arrive at their destination at a particular time, incurring a \emph{scheduling cost} if they arrive earlier or later. 
The total cost of a user is then a combination of the time they spend commuting, and the scheduling cost they incur.
We present a combinatorial algorithm for the natural optimization problem, that of minimizing the average total cost of all users (i.e., maximizing the social welfare).
Based on this, we also show how to set tolls so that this optimal flow is induced as an equilibrium of the underlying game.

\keywords{flows over time \and tolls \and traffic}

\end{abstract}

%% file: introduction.tex
\section{Introduction}\label{sec:intro}
The study of \emph{flows over time} is a classical one in combinatorial optimization; 
it began already with the work of Ford and Fulkerson~\cite{FordFulkerson58} in the 50s.
It is a natural extension of static flows, which associates a single numerical value, representing a total quantity or rate of flow on the arc.
In a flow over time, a second value associated with each arc represents the time it takes for flow to traverse it;
the flow is then described by a function on each arc, representing the rate of flow entering the arc as a function of time.

Classical optimization problems involving static flows
have natural analogs in the flow over time setting (see the surveys~\cite{KRS2009,Skutella2009}).
For example (restricting the discussion to single commodity flows), the \emph{maximum flow over time} problem asks to send as much flow as possible, departing from the source starting from time $0$ and arriving to the sink by a given time horizon $T$;
this can be solved in polynomial time~\cite{FordFulkerson58,FordFulkerson62,FleischerTardos98}.
    A \emph{quickest flow} asks, conversely, for the shortest time horizon necessary to send a given amount of flow.
Of particular importance for us is the notion of an \emph{earliest arrival flow}: this has the very strong property that simultaneously for all $T' \leq T$, the amount of flow arriving by time $T'$ is as large as possible~\cite{Gale}. 
    Such a flow can also be characterized as minimizing the average arrival time~\cite{Jarvis}.
Earliest arrival flows can be ``complicated'', in that they can require exponential space (in the input size) to describe~\cite{Zadeh},  
and determining the average arrival time of an earliest arrival flow is NP-hard \cite{Disser}. 
But they can be constructed in time strongly polynomial in the sum of the input and output size~\cite{Baumann2006}.

Another important aspect of many settings were flow-over-time models are applicable---such as traffic---involves 
game theoretic considerations.
In traffic settings, the flow is made up of a large number of individuals making their own routing choices, 
and aiming to maximize their own utility rather than the overall social welfare (e.g., average journey time).
\emph{Dynamic equilibria}, which is the flow over time equivalent of Wardrop equilibria for static flows, are key objects of study.
Existence, uniqueness, structural and algorithmic issues, and much more have been receiving increasing recent interest from the optimization community~\cite{BFA11,CCL11,CCO17,Correa19,Koch2011,SeringKoch19,SeringSkutella18}.

Traffic, being such a relevant and important topic, has received attention from many different communities, each with their own perspective.
Within the transportation economic literature, modelling other aspects of user choice besides route choice has been considered particularly important. 
A very standard setting, motivated by morning rush-hour traffic, is the following~\cite{Vickrey69,ADL90}.
Users are able to choose not only their route, but also their \emph{departure time}.
They are then concerned not only with their journey time, but also their \emph{arrival time} at the destination.
This is captured in a \emph{scheduling cost function} which we will denote by $\rho$: 
a user arriving at time $\theta$ will experience a scheduling cost of $\rho(\theta)$.
The total disutility of a user is then the sum of their scheduling cost and their journey time (scaled by some factor $\alpha > 0$ representing their value for time spent commuting).
A very standard choice of $\rho$ is
\begin{equation}\label{eq:stdrho}
    \rho(\theta) = \begin{cases}
        -\beta \theta &\text{ if } \theta \leq 0\\
        \phantom{-}\gamma \theta &\text{ if } \theta > 0
    \end{cases},
\end{equation}
where $\beta < \alpha < \gamma$ (it is very bad to be late, but time spent in the office early is better than time spent in traffic).
\ifipco{Our approach can handle essentially general scheduling cost functions, but we will restrict our discussion to strongly unimodal cost functions; these are the most relevant, and this avoids some distracting technical details.}%
{We will allow general scheduling cost functions, though for most of the paper we will focus on strongly unimodal cost functions; 
 these are the most relevant, and this avoids some distracting technical details.}

Two very natural questions can be posed at this point.
The first is a purely optimization question, with no attention paid to the decentralized nature of traffic.

\begin{problem}\label{prob:flow}
    How can one compute a flow over time minimizing the average total cost paid by users, i.e., maximizing the social welfare?
\end{problem}
From now on, we will call a solution to this problem simply an \emph{optimal flow}.

It is well understood that users will typically not coordinate their actions to induce a flow that minimizes total disutility.
There is a huge body of literature (particularly in the setting of static flows~\cite{Roughgardenbook}) investigating this phenomenon. 
In the traffic setting, the relevance of an optimal flow represented by an answer to this question comes primarily via the possibility of \emph{pricing}.
By putting appropriate tolls on roads, we can influence the behaviour of users and the resulting dynamic equilibrium.
Thus:
\begin{problem}\label{prob:tolls}
    How can one set tolls (possibly time-varying) on the arcs of a given instance so that an optimal flow is obtained in dynamic equilibrium?
\end{problem}
One subtlety is that since dynamic equilibria need not be precisely unique, there is a distinction between tolls that induce an optimal flow as \emph{an} equilibrium, compared to tolls for which \emph{all} dynamic equilibria are optimal.
(This is called \emph{weak} and \emph{strong} enforcement by Harks~\cite{Harks} in a general pricing setting.)
We will return to this subtlety shortly.

Questions like these are of great interest to transportation economists.
However, most work in that community has focused on obtaining a fine-grained understanding of very restricted topologies (such as a single link, or multiple parallel links); see \cite{Small2015} for a survey.

Both of these question (for general network topologies) were considered by Yang and Meng~\cite{YangMeng98} 
in a discrete time setting,
by exploiting the notion of \emph{time-expanded graphs}.
This is a standard tool in the area of flows over time; discrete versions all of the optimization questions concerning flows over time mentioned earlier can (in a sense) be dealt with in this way.
A node $v$ in the graph is expanded to a collection $(v, i)$ of nodes, for $i \in \mathbb{Z}$ in a suitable interval, 
and an arc $vw$ of delay $\tau_{vw}$ becomes a collection of arcs $((v,i), (w,i+\tau_{vw}))$ (this assumes a scaling so that $\tau_{vw}$ is a length in multiples of the chosen discrete timesteps).
Scheduling costs are encoded by appropriately setting arc costs from $(t,i)$ to a supersink $t'$ for each $i$, 
and the problem can be solved by a minimum cost static flow computation.
A primary disadvantage of this approach (and in the use of time-expanded graphs more generally) is that the running time of the algorithm depends polynomially on the number of time steps, which can be very large.
Further, it cannot be used to exactly solve the continuous time version (our interest in this paper);
by discretizing time, it can be used to approximate it, but the 
size of the time-expanded graph is inversely proportional to the step size of the discretization.
In the same work~\cite{YangMeng98}, the authors also observe that in the discrete setting, an answer to the second question can be obtained from the time-expanded graph as well.
Taking the LP describing the minimum cost flow problem on the time-expanded graph, the optimal dual solution to this LP provides the necessary tolls to enforce (weakly) an optimal flow.
(This is no big surprise---dual variables can frequently be interpreted as prices.)

\paragraph{An assumption on $\rho$.}
Suppose we consider $\rho$ in the standard form given in \eqref{eq:stdrho}, but with $\beta > \alpha$.
    This means that commuting is considered to be less unpleasant than arriving early.
    A user arriving earlier than time $0$ at the sink would be better off ``waiting'' at the sink before leaving, in order to pay a scheduling cost of $0$.
    Whether waiting in this way is allowed or not depends on the precise way one specifies the model, 
    but it is most natural (and convenient) to allow this.
    If we do so, then it is clear that a scheduling cost function $\rho$ can be replaced by
    \[ 
        \hat{\rho}(\theta) := \min_{\xi \geq \theta} \rho(\xi) + \alpha(\xi - \theta)
    \]
    without changing the optimal flow (except there is no longer any incentive to wait at the sink, and we need not even allow it).
    Then $\theta \to \hat{\rho}(\theta) + \alpha \theta$ is nondecreasing.
    From now on, we always assume that $\rho$ satisfies this; we will call it the \emph{growth bound} on $\rho$.

\paragraph{Our results.}
We give a combinatorial algorithm to compute an optimal flow.
Similarly to the case of earliest arrival flows, this flow can be necessarily complicated, and involves a description length that is 
exponential in the input size.

The algorithm is also similar to that for computing an earliest arrival flow.
It is based on the (possibly exponentially sized) path decomposition of a minimum cost flow into \emph{successive shortest paths}.
In particular, suppose we choose the scheduling cost function to be 
\ifipco{as in \eqref{eq:stdrho}, with $\beta=\alpha$ and $\gamma=\infty$.}%
{\begin{equation}\label{eq:rhoeaf}
    \rho(\theta) = \begin{cases} -\alpha \theta &\text{ if } \theta \leq 0\\
        \infty &\text{ if } \theta > 0
    \end{cases}.
\end{equation}}
    Then the disutility a user experiences is precisely described by how much before time $0$ they depart; all users must arrive by time $0$ to ensure finite cost.
    This is precisely the reversal (both in time and direction of all arcs) of an earliest arrival flow, from the sink to the source.
    Our algorithm will be the same as the earliest arrival flow in this case.
    This also shows that it may be the case that all optimal solutions to Problem~\ref{prob:flow} require exponential size (as a function of the input encoding length), since this is the case for earliest arrival flows.

Despite the close relation to earliest arrival flows, the proof of optimality of our algorithm is rather different.
A key reason for this is the following.
As mentioned, earliest arrival flows have the strong property that the amount of flow arriving before a given deadline $T'$ is the maximum possible, \emph{simultaneously for all choices of $T'$} (up to some maximum depending on the total amount of flow being sent).
This implies that an earliest arrival flow certainly minimizes the average arrival time amongst all possible flows~\cite{Jarvis}, 
but is a substantially stronger property.
A natural analog of this stronger property in our setting would be to ask for a flow for which, simultaneously for any given cost horizon $C' \leq C$, the amount of flow consisting of agents experiencing disutility at most $C'$ is as large as possible.
Unfortunately, in general no such flow exists.
The example is too involved to discuss here, but it relates to some questions on the behaviour of dynamic equilibria in this model that are investigated in a parallel manuscript.

Since the proofs for earliest arrival flows~\cite{Gale,Minieka73,Wilkinson71,Baumann2006}
show this stronger property which does not generalize, we take a different approach.
Our proof is based on duality (of an infinite dimensional LP, though we do not require any technical results on such LPs).
The main technical challenge in our work comes from determining the correct ansatz for the dual solution, as well as exploiting properties of the residual networks obtained from the successive shortest paths algorithm in precisely the right way to demonstrate certain complementary slackness conditions.
As was the case with the time-expanded graph approach, the optimal dual solution immediately provides us with the tolls.
However, we obtain an explicit formula for the optimal tolls, in terms of the successive shortest paths of the graph (see \Cref{sec:algorithm}).
This may be useful in obtaining a better structural understanding of optimal tolls, beyond just their computation.
%
{
    We also remark that a corollary of our result is that there is always an optimal solution without waiting (except at the source).
    \nnote{Added this.}
}

Consider for a moment the model where
users cannot choose their departure time, but instead are released from the source at a fixed rate $u_0$, and simply wish to reach the destination as early as possible.
This is the game-theoretic model that has received the most attention from the flow-over-time perspective~\cite{BFA11,CCO17,Correa19,Koch2011,SeringKoch19}.
Our construction of optimal tolls is applicable to this model as well.
Reverse all arcs, as well as the role of the source and sink (thus making $s$ the new sink),
and also introduce a replacement sink $s'$ and arc $ss'$ of capacity $u_0$ in the original instance .
Then by choosing $\rho$ as described in \ifipco{\eqref{eq:stdrho} with $\beta=\alpha$, $\gamma=\infty$}{\eqref{eq:rhoeaf}}, the optimal flow is an earliest arrival flow, and the tolls we construct will induce it in the original instance (after appropriate time reversal).

\medskip
We now return to the subtlety alluded to earlier: the distinction between strongly enforcing an optimal flow, and only weakly enforcing it.
\ifipco{%
\begin{wrapfigure}{r}{0.48\textwidth}
\vspace*{-0.4cm}
\input{exampleLanesIPCO}
\vspace*{-0.4cm}
\end{wrapfigure}
}{%
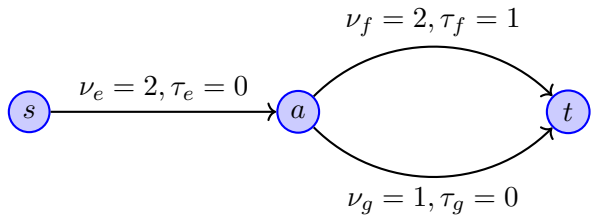
\begin{figure}[H]
\centering
{\input{exampleLanes}}
\caption{An instance where time-varying arc tolls cannot enforce that \emph{all} equilibria are optimal flows.}
\label{fig:nostrong}
\end{figure}
}
Consider the simple instance \ifipco{shown}{in Figure~\ref{fig:nostrong}}.
Suppose that the outflow of arc $a$ is larger than $1$ for some period in the optimum flow, due to the choice of scheduling cost function.
In this period, one unit of flow would take the bottom arc $g$, and the rest will be routed on $f$. 
Since the total cost (including tolls) of all users is the same in a tolled dynamic equilibrium, a toll of cost equivalent to a unit delay on arc $g$ is needed in this period to induce the optimal flow.
But then it will also be an equilibrium to send \emph{all} flow in this period along $f$.

To strongly enforce an optimal flow, we need more flexible tolls.
One way that we can do it is by ``tolling lanes''.
If we are allowed to dynamically divide up the capacity of an arc into ``lanes'' (say a ``fast lane'' and a ``slow lane''), and then separately set time-varying tolls on each lane, then we \emph{can} strongly enforce any optimal flow.
We discuss this further in \Cref{sec:tolls}.
We are not aware of settings where this phenomenon has been previously observed, and it would be interesting to explore this further in a more applied context.

\paragraph{Outline of the paper.}
We introduce some basic notation and notions, as well as a formally define our model, in \Cref{sec:prelim}.
In \Cref{sec:algorithm}, we describe our algorithm, and show that it returns a feasible flow over time; 
we restrict ourselves to the most relevant case of a strictly unimodal scheduling cost function.
In \Cref{sec:optimality} we show optimality of this algorithm, and in \Cref{sec:tolls} we derive optimal tolls from this analysis.
\ipcoskip{Finally, in \Cref{sec:general_scheduling} we remark on the technical changes needed to the algorithm to handle general scheduling cost functions.}

%% file: exampleLanes.tex
\begin{tikzpicture}[remember picture,
inner/.style={circle,draw=blue!50,fill=blue!20,thick,inner sep=3pt}, 
  outer/.style={circle,draw=green,fill=green!20,thick,inner sep=5pt}
  ]
      \node [inner,draw=blue, ] (s) {$s$};
      \node [inner,draw=blue, right=3cm of s] (a) {$a$};
      \node [inner,draw=blue, right=3cm of a] (t) {$t$};

      \draw[black,->,thick] (s) edge[above]  node{$\nu_e = 2, \tau_e = 0$} (a);
  \draw[black,->,thick] (a) edge[above, bend left=45]  node{$\nu_f = 2, \tau_f = 1$} (t);
  \draw[black,->,thick] (a) edge[below, bend right=45]  node{$\nu_g = 1, \tau_g = 0$} (t);

\end{tikzpicture}

%% file: preliminaries.tex
\section{Model and preliminaries}\label{sec:prelim}

The notation $(z)^+$ is used to denote the nonnegative part of $z$, i.e., $(z)^+ = \max\{z,0\}$.
Given $v: X \to \R$ and $A \subseteq X$, we will use the shorthand notation $v(A) := \sum_{a \in A} v(a)$.
We will not distinguish between a map $v: X \to \R$ and a vector in $\R^X$, and so the notation $v_a$ and $v(a)$ is interchangeable.
\ifipco{All graphs will be directed and (purely for notational convenience) simple and without digons.}%
{All graphs considered will be directed.
We assume all graphs to be simple, and that there are no digons (i.e., there are no pairs $v,w \in V$ so that $vw$ and $wv$ are both arcs).
This is for notational convenience only---this restriction can easily be lifted.}

\ipcoskip{We begin with some basic notions and results about static flows and flows over time.
For further details regarding static flows, we refer the reader to the book by Ahuja, Magnanti and Orlin~\cite{AMO93}.
For more about flows over time, we suggest the surveys by Skutella~\cite{Skutella2009} and K\"ohler et al.~\cite{KRS2009}.
}

\paragraph{Static flows.}
Let $G=(V,E)$ be a directed graph, with source node $s \in V$ and sink node $t \in V$.
Each arc $e \in E$ has a \emph{capacity} $\nu_e$ and a \emph{delay} $\tau_e$ (both nonnegative).
We use $\delta^+(v)$ to denote the set of arcs in $E$ with tail $v$, and $\delta^-(v)$ the set of arcs with head $v$.

Consider some $f: E \to \Rplus$ (which we will equivalently view as a vector in $\Rplus^E$).
\ifipco{We use $\nabla f_v$ to denote the net flow into $v \in V$; a (static) $s$-$t$-flow satisfies the usual flow conservation conditions.}{%
For $v \in V$, we define the \emph{net flow} at $v$ (denoted $\nabla f_v$) to be the quantity
\[ \nabla f_v := f(\delta^-(v)) - f(\delta^+(v)) = \sum_{e \in \delta^-(v)} f_e - \sum_{e \in \delta^+(v)} f_e. \]
}
\ipcoskip{
We say that $f$ is a \emph{(static) $s$-$t$-flow of value $Q$} if\ifipco{:}{}
\ifipco{\begin{inparaenum}[(i)]}{\begin{enumerate}[(i)]}
    \item $\nabla f_v = 0$ for all $v \in V \setminus \{s,t\}$, with $\nabla f_t = -\nabla f_s = Q$; and
    \item $f_e \leq \nu_e$ for all $e \in E$.
\ifipco{\end{inparaenum}}{\end{enumerate}}}
Given an $s$-$t$-flow $f$, its \emph{residual network} $G^f = (V, E^f)$ is defined by
\[ E^f = \{ vw : vw \in E \text{ and } f_{vw} < \nu_{vw}\} \cup \{ vw : wv \in E \text{ and } f_{wv} > 0 \}. \]
Call arcs in $E^f \cap E$ \emph{forward arcs} and arcs in $E^f \setminus E$ \emph{backwards arcs}.
The \emph{residual capacity} $\nu^f_e$ of an arc $e \in E^f$ is then $\nu^f_{vw} = \nu_{vw} - f_{vw}$ for $vw$ a forward arc, 
and $\nu^f_{vw} = f_{wv}$ for $vw$ a backwards arc.
We also define $\tau_{vw} = -\tau_{wv}$ for all backwards arcs $vw$.

Given a subset $F \subseteq E$, we use $\chi(P)$ to denote the characteristic vector of $F$.
\ipcoskip{In particular, if $P$ is a path from $v$ to $w$, then $\chi(P)$ is a unit flow from $v$ to $w$.}%
\ifipco{}{\par}%
We make the definitions $\rev{E} := \{ wv : vw \in E\}$ and $\bid{E} := E \cup \rev{E}$.
Given $f,g \in \Rplus^{E}$, we define $f+g$ in the obvious way,
and also define $f-g \in \Rplus^{\bid{E}}$, by interpreting a negative value on $vw$ instead as a positive value on $wv$.

\ipcoskip{
Given a choice of value $Q$, a \emph{minimum cost flow} is an $s$-$t$-flow $f^*$ minimizing $\sum_{e \in E} f^*_e \tau_e$ (amongst all $s$-$t$-flows of value $Q$).
An $s$-$t$-flow $f$ (of the correct value) is a minimum cost flow if and only if $E^{f}$ contains no negative cost cycles, i.e., cycles $C \subseteq E^{f}$ with $\tau(C) < 0$.
}

\paragraph{Flows over time.}
Consider some $f: E \times \R \to \Rplus$. 
We will generally write $f_e(\theta)$ rather than $f(e, \theta)$.
Define the \emph{net flow into $v$ at time $\theta$} by
\[ \nabla f_v(\theta) := \sum_{e \in \delta^-(v)} f_e(\theta - \tau_e) - \sum_{e \in \delta^+(v)} f_e(\theta). \]
Note that $f_e(\theta)$ represents the flow \emph{entering} arc $e$ at time $\theta$; this flow will exit the arc at time $\theta + \tau_e$ (explaining the asymmetry between the terms for flow entering and flow leaving in the above).

We say that $f$ is a \emph{flow over time of value $Q$} if the following hold.
\begin{enumerate}[(i)]
    \item $f$ has compact support (i.e., for some $K$, $f_e(\theta) = 0$ whenever $|\theta| > K$).
    \item $\int_{-\infty}^\infty \nabla f_v(\theta)d\theta = Q(\ind_{v = t} - \ind_{v=s})$ for all $v \in V$.
    \item $\int_{-\infty}^\xi \nabla f_v(\theta)d\theta \geq 0$ for all $v \in V \setminus \{s,t\}$ and $\xi \in \R$. 
    \item $f_e(\theta) \leq \nu_e$ for all $e \in E$ and $\theta \in \R$.
\end{enumerate}
Note that this definition allows for flow to wait at a node; to disallow this and consider only \emph{flows over time without waiting}, we would replace (iii) 
with \ifipco{the condition that $\nabla f_v(\theta) = 0$ for all $v \in V \setminus \{s,t\}$ and  $\theta \in \R$.}%
{the condition
\begin{itemize}[(iii$'$)]
    \item $\nabla f_v(\theta) = 0$ for all $v \in V \setminus \{s,t\}$ and  $\theta \in \R$.
\end{itemize}
}

We also have a natural notion of a residual network in the flow over time setting.
Define, for any flow over time $f$ and $\theta \in \R$, 
\[
    E^f(\theta) = \{ vw : vw \in E \text{ and } f_{vw}(\theta) < \nu_{vw}\} \cup \{ vw : wv \in E \text{ and } f_{wv}(\theta - \tau_{wv}) > 0\}.
\]

\paragraph{Minimizing scheduling cost.}
We are concerned with the following optimization problem.
Given a \emph{scheduling cost function} $\rho: \R \to \Rplus$, as well as a value $\alpha > 0$, determine a flow over time $f$ of value $Q$ that minimizes the sum of the \emph{commute cost} $\alpha \sum_{e \in E} \intR f_e(\theta)d\theta$ and the
\emph{scheduling cost} 
$\intR \nabla f_t(\theta)\cdot \rho(\theta) d\theta$.
As already discussed, we assume that $\rho$ satisfies the growth bound, i.e., that
    $\theta \to \rho(\theta) + \alpha \theta$ is nondecreasing.
This ensures that waiting at $t$ is not needed, which is in fact disallowed by our definition\footnote{Were this really needed, one could simply add a dummy arc $tt'$ to a new sink $t'$.}, and makes various arguments cleaner.
We will also make the assumption that $\rho$ is strongly unimodal\footnote{I.e., (strictly) decreasing until some moment, and then (strictly) increasing.}.
We then assume w.l.o.g.\ that the minimizer of $\rho$ is at $0$, and that $\rho(0) = 0$.
For further technical convenience, by adjusting $\rho$ on a set of measure zero we take 
$\rho$ to be lower semi-continuous.

The unimodal assumption is not necessary; the algorithm and analysis can be extended to essentially general $\rho$, under some very weak technical conditions.
\ifipco{We postpone discussion to the full version of the paper; no major new technical ideas are needed.}%
{We postpone this discussion to the end of the paper.}

We also assume that we are able to query $\rho^{-1}(y)$ for a given rational $y > 0$, obtaining a pair of solutions (one positive, one negative) of moderate bit complexity.
\ipcoskip{Alternatively, we can consider the case where $\rho$ is represented as a piecewise linear function.}

%% file: algorithm.tex
\section{A combinatorial algorithm}\label{sec:algorithm} 

In this section we present an algorithm that computes an optimal flow over time, assuming that $\rho$ is strongly unimodal.
The proof of optimality is discussed in \Cref{sec:optimality}.

We begin by recalling the \emph{successive shortest paths (SSP)} algorithm for computing a minimum cost static flow.
It is not a polynomial time algorithm, 
so it is deficient as an algorithm for static flows, but it provides a structure that is relevant for flows over time.
This is of course well known from its role in constructing earliest arrival flows, which we will briefly detail.

The SSP algorithm construct a sequence of paths $(P_1, P_2, \ldots)$ and associated amounts $(x_1, x_2, \ldots)$ inductively as follows.
Suppose $P_1, \ldots, P_j$ and $x_1, \ldots, x_j$ have been defined.
Let $f^{(j)} = \sum_{i=1}^j x_i \chi(P_i)$,
and let $G_j$ denote the residual graph of $f^{(j)}$.
Also let $d_j(v,w)$ denote the length (w.r.t.\ arc delays $\tau$ in $G_j$) of a shortest path from $v$ to $w$ in $G_j$ (this may be infinite).
By construction, $G_j$ will contain no negative cost cycles, so that $d_j$ is computable.
If $d_j(s,t) = \infty$, we are done; set $m := j$.
Otherwise, define $P_{j+1}$ to be any shortest $s$-$t$-path in $G_j$, and $x_{j+1}$ the minimum capacity in $G_j$ of an arc in $P_{j+1}$.
It can be shown that 
$\sum_{j=1}^r \tilde{x}_j \chi(P_j)$, with $\sum_{j=1}^r \tilde{x}_j = Q$ and $\tilde{x}_j = x_j$ for $j < r$, $0 \leq \tilde{x}_r \leq x_r$, is a minimum cost flow of value $Q$, as long as $Q$ is not larger than the value of a maximum flow.

To construct an earliest arrival flow of value $Q$ and time horizon $T$, 
we (informally) send flow at rate $x_j$ along path $P_j$ for the time interval $[0, T - \tau(P_j)]$, for each $j \in [m]$ (if $\tau(P_j) > T$, we send no flow along the path).
By this, we mean that for each $e=vw \in P_j$, we increase by $x_j$ the value of $f_e(\theta)$ for 
$\theta \in [d_{j-1}(s,v), T - d_{j-1}(v,t)]$ (or if $e$ is a backwards arc, we instead decrease $f_{wv}(\theta)$).
An argument is needed to show that this defines a valid flow, since we must not violate the capacity constraints, and moreover, $P_j$ may contain reverse arcs not present in $G$.

We are now ready to describe our algorithm for minimizing the disutility,
which is a natural variation on the earliest arrival flow algorithm.
It is also constructed from the successive shortest paths, but using a \emph{cost horizon} rather than a \emph{time horizon}.
For now, consider $C$ to be a given value (it will be the ``cost horizon'').
For each $j \in [m]$ with $\alpha d_{j-1}(s,t) \leq C$, we send flow at rate $x_j$ along path $P_j$ for the time interval $[a_j, b_j]$ 
chosen maximally so that 
\ifipco{$\rho(\xi + d_{j-1}(s,t)) \leq C - \alpha d_{j-1}(s,t)$ for all $\xi \in [a_j, b_j]$.}%
    {\[ \rho(\xi + d_{j-1}(s,t)) \leq C - \alpha d_{j-1}(s,t) \quad \text{for all} \quad \xi \in [a_j, b_j]. \]}
    (If $\rho$ is continuous, then of course $\rho(a_j + d_{j-1}(s,t)) = \rho(b_j + d_{j-1}(s,t)) = C - \alpha d_{j-1}(s,t)$.)
Note that a user leaving at time $a_j$ or $b_j$ and using path $P_j$, without waiting at any moment, incurs disutility $C$;
whereas a user leaving at some time $\theta \in (a_j, b_j)$ and using path $P_j$ will incur a strictly smaller total cost.

As we will shortly argue, this results in a feasible flow over time $f$.
Given this, its value will be $\sum_{j=1}^m x_j(b_j - a_j)$.
It is easy to see that this value changes continuously and monotonically with $C$ (here we use the strong unimodality).
Thus a bisection search can be used to determine the correct choice of $C$ for a given value $Q$.
Alternatively, bisection search can be avoided by using Megiddo's parametric search technique~\cite{Megiddo}; this will ensure a strongly polynomial running time, if queries to $\rho^{-1}$ are considered to be of unit cost.

\ipcoskip{
This requires only oracle access to $\rho^{-1}$ (and reasonable control on the bit complexity of $\rho^{-1}(y)$ in terms of $y$).
If $\rho$ is piecewise linear with not too many breakpoints (as is the case, in particular, for the ``standard'' $\beta$/$\gamma$ choice generally used in the transportation economics literature), a third option presents itself.
One can generate the entire parametric curve of $C$ as a function of $Q$, which will also be piecewise linear,
from which the correct choice of $C$ can readily be determined.
This is fairly straightforward, and we omit the details.
}

\paragraph{Feasibility.}
\ipcoskip{In the following we show that the flow resulting flow $f$ is a feasible flow over time. }%
Given a vertex $v\in V$, a time $\theta\in \mathbb{R}$ and $j\in[m]$, let
\[ 
    c_j(v,\theta)=\alpha d_{j-1}(s,t) + \rho(\theta + d_{j-1}(v,t)).
\]
If $v\in P_j$ then $c_j(v,\theta)$ is the travel cost of a user that utilizes path $P_j$ and passes through node $v$ at time $\theta$; 
there does not seem to be a simple interpretation if $v \notin P_j$ however.
Now define
\begin{equation}
    J(v,\theta) = \max \{ j \in [m] : c_j(v,\theta) \leq C \},
\end{equation}
with the convention that the maximum over the empty set is $0$.
The motivation for this definition comes from the following theorem, which completely characterizes $f$.
\ipcoskip
{(If preferred, one could even think of this theorem as providing the definition of $f$.)}
\begin{theorem}\label{thm:fvalue}
    $f_{vw}(\theta) = f_{vw}^{(J(v,\theta))}$ for any $vw \in E$ and $\theta \in \R$.
\end{theorem}
\ifipco{Since $f$ has value $Q$ and satisfies flow conservation by construction, the feasibility of $f$ is an immediate corollary of this theorem.}%
{\begin{corollary}
    $f$ is a feasible flow over time. \nnote{Possibly mention it has no waiting.}
\end{corollary}
\begin{proof}
    By the way that we constructed $f$, it has value $Q$ and satisfies flow conservation. 
    Only nonnegativity and the capacity constraint remain, which follows from the theorem.
\end{proof}
}
\ifipco{We sketch the proof in the appendix.}{%
Before proving \Cref{thm:fvalue}, we need the following lemma.
\begin{lemma}\label{lem:monotonicityStraight}
    $c_j(v,\theta)$ is nondecreasing with $j$ for any $\theta \in \R$. 
\end{lemma}
\begin{proof}   
Consider any $j \in [m-1]$; we show that $c_{j+1}(v,\theta) \geq c_j(v,\theta)$.
Suppose $Q$ is a shortest $v$-$t$-path in $G_{j-1}$, so $\tau(Q) = d_{j-1}(v,t)$. 
Consider the unit $v$-$t$ flow $g = \chi(P_{j+1}) - \chi(P_j) + \chi(Q)$ in $\bid{E}$.
Now observe that the support of $g$ is contained in $G_{j}$:
$P_{j+1}$ and $\rev{P_j}$ are certainly contained in $G_j$; and if $e \in Q \cap (E_{j-1} \setminus E_j)$, then $e \in P_j$.
Since $G_j$ contains no negative cost cycles, the cost of $g$ is at least that of a shortest $v$-$t$-path in $G_j$, and so
$d_{j}(v,t) \leq  \tau(P_{j+1})- \tau(P_{j}) +\tau(Q)$.
Finally, we can conclude 
\begin{align*}
\alpha d_{j}(s,t) +\rho(\theta+d_{j}(v,t)) 
&= \alpha d_{j}(s,t) +\rho(\theta+d_{j-1}(v,t)) - \rho(\theta+d_{j-1}(v,t))+\rho(\theta+d_{j}(v,t))\\
&\geq \alpha d_{j}(s,t) +\rho(\theta+d_{j-1}(v,t)) - \alpha(d_{j}(v,t)-d_{j-1}(v,t))\\
&\geq \alpha d_{j-1}(s,t) +\rho(\theta+d_{j-1}(v,t)),
\end{align*}
where the first inequality follows from the growth assumption, using $d_{j}(v,t)\geq d_{j-1}(v,t)$.
\end{proof}

\begin{proofof}{\Cref{thm:fvalue}}
Fix some $vw \in E$ and $\theta \in \R$.
Consider now any $P_j$ (with $\alpha \tau(P_j) \leq C$, so that it is used for a nontrivial interval), with $vw \in P_j$.
Since $P_j$ is a shortest path in $G_{j-1}$, 
if we send flow along this path starting from some time $\xi$, it will arrive at $v$ at
time $\xi + d_{j-1}(s,v)$.
Considering the definition of the interval $[a_j, b_j]$, we see that $P_j$ contributes flow to $vw$ at time $\theta$ if $c_j(v,\theta) \leq C$.
By \Cref{lem:monotonicityStraight}, this occurs precisely if $j \leq J(v,\theta)$.

Considering in similar fashion paths $P_j$ with $wv \in P_j$ (and noting that $J(w,\theta + \tau_{vw}) = J(v,\theta)$), we determine that
\[
   f_{vw}(\theta) = \sum_{\substack{j: vw \in P_j\\ j \leq J(v,\theta)}} x_j - 
         \sum_{\substack{j: wv \in P_j\\ j \leq J(v,\theta)}} x_j 
       = f^{(J(v,\theta))}_{vw}. \qedhere
\]
\end{proofof}
}

%% file: duality.tex
\section{Optimality}\label{sec:optimality}
\ipcoskip{In this section, we show that our proposed algorithm does return an optimal flow.}

\ifipco{\paragraph{Duality-based certificates of optimality.}}{\subsection{Duality-based certificates of optimality}\label{sec:duality}}

    We can write the problem we are interested in as a (doubly) infinite linear program as follows:
\ifipco{
\medskip

\noindent
\begin{tabular}{l r l l l r}
$ \min\quad$ & \multicolumn{5}{l}{
$ 
\int_{-\infty}^{\infty} \rho(\theta) \nabla f_{t}(\theta) d\theta 
+ \alpha \sum_{e\in E}\tau_e\int_{-\infty}^{\infty} f_e(\theta) d\theta 
+ \alpha \sum_{v\in V\setminus\{s,t\}}\int_{-\infty}^{\infty} z_v(\theta) d\theta 
$}\\[0.21cm] 
$ \text{s.t.}$ & \ifipco{$ 
 -\int_{-\infty}^{\infty} \nabla f_{s}(\theta) d\theta = \int_{-\infty}^{\infty} \nabla f_{t}(\theta) d\theta $ & $ =$ & $ Q$}
{
$ \int_{-\infty}^{\infty} \nabla f_{s}(\theta) d\theta$ & $ =$ & $ -Q$}  
& &\multirow{3}{*}{\quad\tagarray \label{eq:primal}}\\[0.2cm]
\ifipco{}{
& $ \int_{-\infty}^{\infty} \nabla f_{t}(\theta) d\theta$ & $ =$ & $ Q$ 
\\[0.2cm]}
& $ \int_{-\infty}^{\theta} \nabla f_v(\xi) d\xi $ 
& $ =$ & $ z_v(\theta)$ & $ \qquad\forall v\in V\setminus \{s,t\}, \theta \in \R$&\\[0.2cm]
& $ f_e(\theta)$ 
& $ \leq $ & $ \nu_e$ & $\qquad\forall e\in E, \theta \in \R$\\[0.2cm]
& $ z,f$ 
& $ \geq $ & $ 0$ & &\\[0.2cm]
\end{tabular}
}{
\begin{equation}\label{eq:primal} 
\begin{aligned}
\min \qquad 
\lpint_{-\infty}^{\infty} \rho(\theta) \nabla f_{t}(\theta) d\theta 
&+ \alpha \lpsum_{e\in E}\tau_e\lpint_{-\infty}^{\infty} f_e(\theta) d\theta 
+ \alpha \lpsum_{v\in V\setminus\{s,t\}}\lpint_{-\infty}^{\infty} z_v(\theta) d\theta 
\\
\text{s.t.} \qquad
\ifipco{
    -\lpint_{-\infty}^{\infty} \nabla f_{s}(\theta) d\theta &= \lpint_{-\infty}^{\infty} \nabla f_{t}(\theta) d\theta = Q\\}%
{    \lpint_{-\infty}^{\infty} \nabla f_{s}(\theta) d\theta &= -Q \\
\lpint_{-\infty}^{\infty} \nabla f_{t}(\theta) d\theta &= Q\\}
    \lpint_{-\infty}^{\theta} \nabla f_v(\xi) d\xi &=  z_v(\theta) \qquad \forall v\in V\setminus \{s,t\}, \theta \in \R\\
   f_e(\theta) &\leq \nu_e \qquad \forall e\in E, \theta \in \R \\
  z, f &\geq 0
\end{aligned}
\end{equation}
}

\noindent
Here, $z_v(\theta)$ represents the amount of flow waiting at node $v$ at time $\theta$ (which must always be nonnegative).
The travel cost is captured on a per-arc basis, including waiting time as well.

The following theorem provides a certificate of optimality of a feasible solution to \eqref{eq:primal}.
\begin{theorem}\label{thm:optimality}
    Let $f$ be a flow over time with value $Q$, and 
    suppose that $\pi: V \times \R \to \R$ satisfies the following, for some choice of $C$:
    \begin{enumerate}[(i)]
        \item \label{prop:noninc} $\theta \to \pi_v(\theta) - \alpha \theta$ is nonincreasing.
        \item \label{prop:piedge} $\pi_w(\theta + \tau_{vw}) \leq \pi_v(\theta) + \alpha \tau_{vw}$ for all $\theta \in \R, vw \in E^f(\theta)$.
        \item \label{prop:pis}
    $\pi_s(\theta) = 0$ for all $\theta \in \R$.
\item \label{prop:pit} $\pi_t(\theta) = (C - \rho(\theta))^+$ for all $\theta \in \R$, and $\nabla f_t(\theta) = 0$ whenever $\rho(\theta) > C$.
    \end{enumerate}
Then $f$ is an optimal solution.
\end{theorem}

Essentially, $\pi_v(\theta)$ are dual variables, and the assumptions of the theorem are that $f$ and $\pi$ satisfy the complementary slackness conditions.
There are many extensions of LP duality theory to infinite dimensional settings, e.g., \cite{Grinold,RSB92};
however the situation is subtle, since strong duality and even weak duality can fail~\cite{RSB92}.
\ifipco{We prefer to avoid technicalities and derive it directly (the proof is given in the full version).}%
{We prefer to avoid technicalities and derive it directly.}

\ipcoskip{
\begin{proofof}{\Cref{thm:optimality}}
    We will need the following technical lemma (obvious via integrating by parts in the case that $h$ is also absolutely continuous). 
    \begin{claim}\label{claim:decint}
    Let $h: \R \to \R$ be a nonincreasing function, and $z: \R \to \Rplus$ be an absolutely continuous nonnegative function with compact support.
    Then $\intR h(\theta)z'(\theta) d\theta \geq 0$.
\end{claim}
\begin{proof}
    Since $\intR z'(\theta)d\theta = 0$, we may assume by shifting if necessary that $h$ is nonnegative on the support of $z$.
    Let $\mu$ be a measure so that $\mu([\theta, \infty)) = h(\theta)$ for almost every $\theta$ in the support of $z$.
    We certainly have that for any $\theta$,
    \[
        \int_{-\infty}^{\theta} z'(\xi)d\xi = \bigl[z(\xi)\bigr]^{\theta}_{-\infty} = z(\theta) - 0 \geq 0. 
    \]
    Thus
    \begin{align*}
        \intR \intR \ind_{\xi \leq \theta} z'(\xi)d\xi d\mu(\theta) \geq 0,
    \end{align*}
    from which we obtain the result by applying Fubini's theorem.
\end{proof}

Define, for each $vw \in E$,
\[ 
    \delta_{vw}(\theta) = (\pi_w(\theta + \tau_{vw}) - \pi_v(\theta) - \alpha \tau_{vw})^+.
\]

Now let $g, \zg$ be any feasible solution to \eqref{eq:primal} with compact support. Consider any $v \in V \setminus \{s,t\}$, and observe that
\begin{equation}\label{eq:nonnegterm}
    \intR \bigl(\pi_v(\theta) \nabla g_v(\theta) + \alpha \zg_v(\theta)\bigr)d\theta 
    = \intR \bigl( (\pi_v(\theta) - \alpha \theta)\nabla g_v(\theta)\bigr)d\theta \;+\; \Bigl[\alpha \theta \zg_v(\theta)\Bigr]_{-\infty}^{\infty}
    \geq 0,
\end{equation}
by the above claim, exploiting property (\ref{prop:noninc}).

We have
\begin{align*}
    \cost(g) &= \intR \rho(\theta) \nabla g_{t}(\theta) d\theta 
+ \alpha \sum_{e\in E}\intR \tau_e g_e(\theta) d\theta 
+ \alpha \sum_{v\in V\setminus\{s,t\}}\intR \zg_v(\theta) d\theta \\
             &\overset{(*)}{\geq} \intR (C - \pi_t(\theta)) \nabla g_t(\theta)d\theta + \alpha \sum_{vw \in E} \intR \bigl(\pi_w(\theta + \tau_{vw}) - \pi_v(\theta) - \delta_{vw}(\theta)\bigr) g_e(\theta) d\theta\\
             &\qquad +  \alpha \sum_{v\in V\setminus\{s,t\}}\intR \zg_v(\theta) d\theta \\
             &\overset{(**)}{=}
                 CQ + \sum_{v \in V \setminus \{s,t\}} \intR [\pi_v(\theta)\nabla g_v(\theta) + \alpha \zg_v(\theta)]d\theta - \sum_{e \in E} \delta_e(\theta) g_e(\theta)\\
             &\overset{(***)}{\geq} CQ - \sum_{e \in E} \intR \delta_e(\theta) \nu_e d\theta.
\end{align*}
In the above, ($*$) holds by the definitions of $\pi_t$ and $\delta_e$; ($**$) follows by recombining the $g_e(\theta)$ terms and recalling that $\pi_s \equiv 0$ and that $g$ has value $Q$;
and ($*\!*\!*$) follows from \eqref{eq:nonnegterm}, and the inequalities $\delta_e(\theta) \geq 0$ and $g_e(\theta) \leq \nu_e$ that hold for all $e \in E$ and $\theta \in \R$.

Finally, observe that all of the inequalities in the above hold with equality if $g=f$.
Property (\ref{prop:piedge}) implies that if $f_{vw}(\theta) > 0$ (so that $wv \in E^f(\theta)$), then 
$\delta_{vw}(\theta) = \pi_w(\theta + \tau_{vw}) - \pi_v(\theta) - \alpha \tau_{vw}$, yielding equality in ($*$).
It also implies that if $f_{vw}(\theta) < \nu_{vw}$ (so that $vw \in E^f(\theta)$) then $\delta_{vw}(\theta) = 0$, yielding equality in ($*\!*\!*$).
\end{proofof}
}

\ipcoskip{As is often the case, the optimal dual solution also provides us the prescription for tolls to induce the optimum flow.
We delay this discussion to \Cref{sec:tolls}.}

%% file: analysisNew.tex
\ifipco{\paragraph{The dual prescription.}}{\subsection{The dual prescription}\label{sec:prescription}}
We now give a certificate of optimality $\pi: V \times \R \to \R$ for \eqref{eq:primal} that satisfies the conditions of \ifipco{the above LP}{\Cref{thm:optimality}}. 
Given a vertex $v\in V$ and a time $\theta\in \mathbb{R}$
let 
\ifipco{%
    \begin{flalign*}
        \pi_v(\theta)&=\max\{\pi'_v(\theta), \bar{\pi}_v(\theta), 0\} &\\[0.3em]
    \text{where} \qquad \qquad
   \pi'_v(\theta) &= - \alpha d_{J(v,\theta)}(v,s),\\
   \bar{\pi}_v(\theta) &=C-\alpha d_{J(v,\theta)}(v,t) - \rho (\theta+d_{J(v,\theta)}(v,t)).
\end{flalign*}
}{%
\[\pi_v(\theta)=\max\{\pi'_v(\theta), \bar{\pi}_v(\theta), 0\} \] 
where
\begin{align*}
   \pi'_v(\theta) &= - \alpha d_{J(v,\theta)}(v,s),\\
   \bar{\pi}_v(\theta) &=C-\alpha d_{J(v,\theta)}(v,t) - \rho (\theta+d_{J(v,\theta)}(v,t)).
\end{align*}
}
Notice that $\pi_s(\theta)=0$ and $\pi_t(\theta)=\max\{C-\rho(\theta),0\}$ for all $\theta\in \R$ and thus conditions (\ref{prop:pis}) and (\ref{prop:pit}) of \Cref{thm:optimality} hold.
\ifipco{
        The bulk of the technical work is in showing the remaining conditions; we sketch some part of the proof in the appendix.
    }%
{For the remaining conditions, we begin with some basic facts about distance labels associated with successive shortest paths
(statements of a similar flavour can be found in Ahuja et al.~\cite{AMO93}, for example).}

\ipcoskip{
\begin{lemma}\label{lem:vsdec}
$d_j(v,s)$ is nonincreasing with $j$, for every $v \in V$.
\end{lemma}
\begin{proof}
We show that $d_{j-1}(v,s) \geq d_j(v,s)$ for any $j \in [m]$. 
Let $Q$ be a shortest $v$-$s$-path in $G_{j-1}$ 
(if there is no such path, then there is nothing to prove). 
If $P_j \cap Q = \emptyset$, then $Q \subseteq E_j$, so the claim holds. 
Otherwise, let $xy$ be the
last (i.e., closest to $s$) edge of $Q$ that is also in $P_j$, 
and let $Q_y$ denote the subpath of $Q$ from $v$ to $y$. 
Also let $P_y$ denote the subpath of $P_j$ from $y$ to $v$. 
Then $C := Q_y + P_y$ is a directed cycle 
(or collection of cycles, possibly with some edges included twice). 
Since $C \subseteq E_{j-1}$,
$\tau(C) \geq 0$, so $\tau(\rev{P}_y) \leq \tau(Q_y)$.

Define $Q'$ to be the path obtained by appending the subpath of $Q$
from $y$ to $s$ to $\rev{P}_y$. 
Then $Q' \subseteq E_j$, and $\tau(Q') \leq \tau(Q)$, as required.
\end{proof}

\begin{lemma}\label{lem:equalityDario}
    For all $j \in [m]$ and $v \in V$, $d_{j-1}(v,t)-d_{j-1}(s,t)=d_{j}(v,s)$.
\end{lemma}
\begin{proof}
To show that $d_{j-1}(v,t)-d_{j-1}(s,t)\leq d_{j}(v,s)$, 
let $Q$ be a shortest $v$-$s$-path in $G_j$, 
and let $Q'$ be a $v$-$t$-path contained in $P_j + Q$ (arcs in opposite directions are cancelled).
Then $Q'$ is in $G_{j-1}$; if $e$ is an arc in $Q$ not in $G_{j-1}$, 
then $e$ is the reverse of an arc of $P_j$, and hence not in $P_j + Q$. 
So $d_{j-1}(v,t) \leq \tau(Q')$.
But since $G_{j-1}$ has no negative cost cycles, $\tau(Q') \leq \tau(Q) + \tau(P_j)$. 

To show that $d_{j-1}(v,t)-d_{j-1}(s,t)\geq d_{j}(v,s)$, let $\bar{Q}$ be a shortest $v$-$t$-path in $G_{j-1}$.
Let $w$ be the first (i.e., closest to $v$) vertex present 
in both $\bar{Q}$ and $P_j$ (notice that $w$ might be equal to $v$ or $t$) and let $Q$ be the $v$-$w$-path contained in $\bar{Q}$.
Then 
\begin{equation}\label{eq:Qchoice}
    d_{j-1}(v,t)=d_{j-1}(w,t)+\tau(Q).
\end{equation}
Since $Q\subseteq E_{j}$, we have that: 
\begin{align*}
d_{j}(v,s)&\leq d_{j}(w,s)+\tau(Q)\\
 &=d_{j-1}(w,t)-d_{j-1}(s,t)+\tau(Q) &&\text{since $w \in P_j$}\\
 &=d_{j-1}(v,t)-d_{j-1}(s,t) &&\text{by \eqref{eq:Qchoice}}.
\end{align*}
This concludes the proof.
\end{proof}

Now we are ready to show that $\pi$ satisfies  conditions (\ref{prop:noninc}) and (\ref{prop:piedge}) of \Cref{thm:optimality}.
\begin{lemma}\label{lem:potential_plus_epsilon}
$\theta \to \pi_v(\theta) - \alpha \theta$ is nonincreasing.
\end{lemma}
\begin{proof}
    Fix any $\theta \in \R$ and $\epsilon \geq 0$.
    We show that $\pi_v(\theta) \geq \pi_v(\theta+\epsilon) - \alpha \epsilon$.
Let $j := J(v,\theta)$ and $\ell := J(v,\theta + \epsilon)$.

\begin{itemize}
\item
  \textbf{Case 1:} $\pi_v(\theta+\epsilon) =  - \alpha d_{\ell}(v,s)$.

  If $\ell \leq j$, then by \Cref{lem:vsdec}
  \[ \pi_v(\theta) 
    \geq \pi'_v(\theta) 
    =  - \alpha d_{j}(v,s) 	
    \geq  - \alpha d_{\ell}(v,s)
= \pi_v(\theta+\epsilon). \]
  So suppose $\ell > j$. By the definition of $J(v,\theta+\epsilon)$, we know that 
    \begin{equation}\label{eq:CgreaterStraight}
        \alpha d_{\ell-1}(s,t) + \rho(\theta+\epsilon + d_{\ell-1}(v,t))\leq C.
    \end{equation}
  As a consequence, we have that: 
  \begin{align*}
    \pi_v(\theta) &\geq \bar{\pi}_v(\theta) \\
     &= C-\alpha d_{j}(v,t) - \rho (\theta+d_{j}(v,t))\\
      &\overset{(*)}{\geq} C-\alpha d_{j}(v,t) - \rho (\theta+\epsilon+d_{\ell-1}(v,t)) - \alpha \left(\epsilon+d_{\ell-1}(v,t)-d_{j}(v,t)\right)\\
     &\geq \alpha d_{\ell-1}(s,t) - \alpha \epsilon - \alpha d_{\ell-1}(v,t) && \text{by (\ref{eq:CgreaterStraight})}\\
     &= - \alpha \epsilon -\alpha d_{\ell}(v,s) && \text{by \Cref{lem:equalityDario}} \\
     &= \pi_v(\theta+\epsilon)-\alpha \epsilon. 
  \end{align*}
  Inequality $(*)$ follows from the growth assumption on $\rho$ combined with the fact that $\theta+\epsilon+d_{\ell-1}(v,t)\geq \theta+d_{j}(v,t)$.
\item
  \textbf{Case 2:} 
  $\pi_v(\theta+\epsilon) = C-\alpha d_{\ell}(v,t) - \rho (\theta+\epsilon+d_{\ell}(v,t))$.

  If $\ell\geq j$, then:
  \begin{align*}
    \pi_v(\theta) &\geq \bar{\pi}_v(\theta) \\
     &= C-\alpha d_{j}(v,t) - \rho (\theta+\epsilon+d_{j}(v,t))\\
     &= C-\alpha d_{j}(v,t) - \rho (\theta+\epsilon+d_{\ell}(v,t)) + \rho (\theta+\epsilon+d_{\ell}(v,t))- \rho (\theta+d_{j}(v,t))\\
     &\geq C-\alpha d_{j}(v,t) - \rho (\theta+\epsilon+d_{\ell}(v,t)) - \alpha \left(\epsilon+d_{\ell}(v,t)-d_{j}(v,t)\right)\\
    &= C- \rho (\theta+\epsilon+d_{\ell}(v,t)) - \alpha\epsilon -\alpha d_{\ell}(v,t)\\
    &= \pi_v(\theta+\epsilon) - \alpha \epsilon.
  \end{align*}
    The second inequality follows again from the growth assumption, this time combined with the inequality $d_{\ell}(v,t)\geq d_{j}(v,t)$.

  If $\ell<j$, by definition of $J(v,\theta+\epsilon)$  we have that 
      \[
    \alpha d_{\ell}(s,t) + \rho(\theta+\epsilon+ d_{\ell}(v,t))>C.
    \] 
    From this, we obtain
  \begin{align*}
    \pi_v(\theta ) 
    &\geq \pi'_v(\theta) \\
    &= - \alpha d_{j}(v,s)	\\
    &> C-\alpha d_{\ell}(s,t) - \rho(\theta +\epsilon+ d_{\ell}(v,t)) - \alpha d_{j}(v,s) 	\\
    &\geq C-\alpha d_{\ell}(s,t) - \rho(\theta +\epsilon+ d_{\ell}(v,t)) - \alpha d_{\ell+1}(v,s) &&\text{by \Cref{lem:vsdec}}\\
    &= C-\alpha d_{\ell}(v,t) - \rho(\theta +\epsilon+ d_{\ell}(v,t))  &&\text{by \Cref{lem:equalityDario}}\\
    &= \pi_v(\theta).
  \end{align*} 
\end{itemize}
\end{proof}

\begin{lemma}\label{lem:potential_plus_tau}
    If $vw \in E^f(\theta)$, then
$\pi_w(\theta + \tau_{vw}) \leq \pi_v(\theta) + \alpha \tau_{vw}$.
\end{lemma}
\begin{proof}
Let $j := J(v,\theta)$ and $\ell := J(w,\theta + \tau_{vw})$.
Note that since $vw \in E^f(\theta)$, \Cref{thm:fvalue} implies that $vw \in E_j$.
\begin{itemize}
\item
  \textbf{Case 1:} $\pi_w(\theta+\tau_{vw}) = - \alpha d_{\ell}(w,s)$.
  
  If $\ell \leq j$, then 
    \begin{align*}
    \pi_v(\theta ) 
    &\geq  - \alpha d_j(v,s) 	\\
    &\geq  - \alpha \tau_{vw} - \alpha d_{j}(w,s) && \text{since } vw \in E_j\\
    &\geq  - \alpha \tau_{vw} - \alpha d_{\ell}(w,s) &&\text{by \Cref{lem:vsdec}}\\
    &= \pi_w(\theta+ \tau_{vw}) - \alpha \tau_{vw}.
  \end{align*} 
  
    So suppose $\ell >j$. 
    By the definition of $J(w,\theta + \tau_{vw})$ we know that 
    \begin{equation}\label{eq:Cgreater2Straight}
    \alpha d_{\ell-1}(s,t) + \rho(\theta +\tau_{vw}+ d_{\ell-1}(w,t))\leq C.
    \end{equation}
    Since $vw \in E_j$ and $d_{j}(w,t)\leq d_{\ell-1}(w,t)$, we also have
    \begin{equation}\label{eq:intPosStraight}
    \theta+d_{j}(v,t)\leq \theta+\tau_{vw}+d_{\ell-1}(w,t).
    \end{equation}
    Thus
  \begin{align*}
    \pi_v(\theta) 
                  &\geq C-\alpha d_{j}(v,t) - \rho (\theta+d_{j}(v,t)) \\ 
     &\geq C-\alpha d_{j}(v,t) - \rho (\theta+\tau_{vw}+d_{\ell-1}(w,t)) - \alpha \left(\tau_{vw}+d_{\ell-1}(w,t)-d_{j}(v,t)\right)\\
     &\geq \alpha d_{\ell-1}(s,t) - \alpha \tau_{vw} - \alpha d_{\ell-1}(w,t) && \text{by (\ref{eq:Cgreater2Straight})}\\
     &= - \alpha \tau_{vw} -\alpha d_{\ell}(w,s) && \text{by \Cref{lem:equalityDario}} \\
     &= \pi_w(\theta+\tau_{vw})-\alpha \tau_{vw} 
  \end{align*}
   where the second inequality follows from the growth assumption and from (\ref{eq:intPosStraight}).

\item
  \textbf{Case 2:} 
  $\pi_w(\theta+\tau_{vw}) = C-\alpha d_{\ell}(w,t) - \rho (\theta+\tau_{vw}+d_{\ell}(w,t))$.

  If $\ell\geq j$, since $vw \in E_j$ and $d_{j}(w,t)\leq d_{\ell}(w,t)$, we have that
    \begin{equation}\label{eq:intPosStraight2}
    \theta+d_{j}(v,t)\leq \theta+\tau_{vw}+d_{\ell}(w,t).
    \end{equation}
    As a consequence, exploiting also the growth assumption, we have
  \begin{align*}
    \pi_v(\theta) &\geq \bar{\pi}_v(\theta) \\
     &= C-\alpha d_{j}(v,t) - \rho (\theta+d_{j}(v,t))\\
     &= C-\alpha d_{j}(v,t) - \rho (\theta+\tau_{vw}+d_{\ell}(w,t)) + \rho (\theta+\tau_{vw}+d_{\ell}(w,t))- \rho (\theta+d_{j}(v,t))\\
     &\geq C-\alpha d_{j}(v,t) - \rho (\theta+\tau_{vw}+d_{\ell}(w,t)) - \alpha \left(\tau_{vw}+d_{\ell}(w,t)-d_{j}(v,t)\right)\\
     &= C- \rho (\theta+\tau_{vw}+d_{\ell}(w,t)) - \alpha\tau_{vw} -\alpha d_{\ell}(w,t)\\
     &= \pi_w(\theta+\tau_{vw})-\alpha \tau_{vw}.
  \end{align*}

  If $\ell<j$, by definition of $J(w,\theta+\tau_{vw})$  we have that 
    \begin{equation}\label{eq:intPosStraight3}
    \alpha d_{\ell}(s,t) + \rho(\theta+\tau_{vw}+ d_{\ell}(w,t))>C.
    \end{equation}
    Thus
  \begin{align*}
    \pi_v(\theta ) 
    &\geq - \alpha d_{j}(v,s)\\
    &\geq  - \alpha d_{j}(w,s)-\alpha\tau_{vw} 	 && \text{since } vw \in E_j\\
    &\geq  - \alpha d_{\ell+1}(w,s)-\alpha\tau_{vw} 	 && \text{by \Cref{lem:vsdec}}\\
    &> C-\alpha d_{\ell}(s,t) - \rho(\theta +\tau_{vw}+ d_{\ell}(w,t)) - \alpha d_{\ell+1}(w,s)-\alpha\tau_{vw}	 &&\text{by  (\ref{eq:intPosStraight3})}	\\
    &= C-\alpha d_{\ell}(w,t) - \rho(\theta +\tau_{vw}+ d_{\ell}(w,t))-\alpha\tau_{vw}  &&\text{by \Cref{lem:equalityDario}}\\
    &= \pi_w(\theta+\tau_{vw})-\alpha\tau_{vw}.
  \end{align*} 
\end{itemize}
\end{proof}
}

%% file: tolls.tex
\section{Optimal tolls}\label{sec:tolls}

Tolls $\delta: E \times \R \to \Rplus$ are per-arc, time-varying and nonnegative.
The value $\delta_e(\xi)$ represents the toll a user is charged upon entering the link at time $\xi$.

We have the following theorem.

\begin{theorem}
    Let $(f,\pi)$ be an optimal primal-dual solution to \eqref{eq:primal} (as constructed in \Cref{sec:algorithm} and \Cref{sec:optimality}) and define, for each $vw \in E$,
\[ 
    \delta_{vw}(\theta) = (\pi_w(\theta + \tau_{vw}) - \pi_v(\theta) - \alpha \tau_{vw})^+.
\]
Then $f$ is a dynamic equilibrium under tolls $\delta$.
\end{theorem}

Of course, to make sense of this theorem we must know what is meant by a dynamic equilibrium under tolls.
A precise definition requires introducing the full game-theoretic fluid queueing model (also known as the Vickrey bottleneck model)~\cite{Vickrey69,Koch2011}.
Tolls and departure time choice can be introduced into the definition of a dynamic equilibrium discussed in these works.
Rather than going this route, we will show that the tolls satisfy a strong property that very clearly ensures the equilibrium property.

\ifipco{We show (in the full version---it is straightforward) that the following holds.}{We show that the following holds.}
A user starting from some $v \in V$ at some time $\theta \in \R$ cannot incur a total cost (including scheduling cost, and tolls and commuting cost from this point forward) less than $C - \pi_v(\theta)$. This is even allowing the user to take any link at any time, as if no other users were present in the network.
Since the flow represents a solution where all users incur a total cost of precisely $C$, this must certainly be an equilibrium.

\ipcoskip{
To see this, consider any $s$-$t$-path $P$ in $E$ and arrival times $\theta_v$ for each $v \in P$ valid for this path; 
so $\theta_w \geq \theta_v + \tau_{vw}$ for every $vw \in P$.
Thus by properties (\ref{prop:noninc}) and (\ref{prop:piedge}) of \Cref{thm:optimality}, $\pi_w(\theta_w) - \pi_v(\theta_v) \leq \alpha(\theta_w - \theta_v)$, with equality if $wv \in E^f(\theta_w)$ and $\theta_w = \theta_v + \tau_{vw}$.
Then the cost of a user using this route is
\begin{align*}
    &\phantom{=} \rho(\theta_t) + \sum_{e=vw \in P} [\alpha \tau_e + \delta_e(\theta_v)]\\
    &\geq \rho(\theta_t) + \pi_t(\theta_t) - \pi_s(\theta_s)\\
    &= \rho(\theta_t) + (C - \rho(\theta_t))^+\\
    &\geq C.
\end{align*}
The inequalities are all tight if for all $vw \in P$, $\theta_w = \theta_v + \tau_{vw}$ and $f_{vw}(\theta_v) > 0$, by the previous observations as well as property (\ref{prop:pit}). 
So if the aggregate choices of the users are described by $f$, all users pay exactly $C$.
}

\medskip

As already discussed, we cannot in general strongly enforce an optimal flow.
The following shows that the ``lane tolling'' approach suffices to do this.
\begin{theorem}
    With $f, \pi$ and $\delta$ as in the previous theorem, any dynamic equilibrium $g$ satisfying $g_e(\theta) \leq f_e(\theta)$ for all $e \in E$, $\theta \in \R$ is $g=f$.
\end{theorem}
Essentially, being able to dynamically split and separately toll the capacity of a link allows us to easily rule out all other potential equilibria just by using tolls to artificially constrict the capacities (in addition to choosing tolls that weakly enforce the desired flow, which is still needed).
Tolling in this way seems quite distant from what could be imaginable in realistic traffic scenarios. 
But it does raise the interesting question of whether there is a tolling scheme which can strongly enforce an optimum flow,
but which is more restricted (and more plausible) than fully dynamic lane tolling.
Another natural question would be to determine if an optimum flow can be strongly enforced using lane tolling only on certain specified edges.
We leave these as open questions.

%% file: generalCase.tex
\section{General scheduling costs}\label{sec:general_scheduling}
We now briefly discuss general scheduling costs, satisfying only the growth assumption as well as minor technical conditions.
We will not give full details, but just highlight what changes need to be made in the algorithm and analysis.

First, suppose that $\rho$ was unimodal, but not strongly unimodal.
This introduces the complication that we cannot uniquely invert $\rho$ on $[0, \infty)$ and $(-\infty, 0]$; and also that $\rho^{-1}$ is discontinuous.
If we use the prescription of the interval $[a_j, b_j]$ given in \Cref{sec:algorithm}, namely the maximal interval for which
\begin{equation}\label{eq:interval_def}
\rho(\xi + d_{j-1}(s,t)) \leq C - \alpha d_{j-1}(s,t) \quad \text{for all} \quad \xi \in [a_j, b_j],
\end{equation}
we run into the difficulty that $Q$ is no longer a continuous function of $C$.
Thus we may find a choice $C_0$ so that the resulting amount of flow is strictly less than $Q$, but the amount of flow corresponding to $C_0 + \epsilon$ is strictly larger than $Q$, for any positive $\epsilon$.

To remedy this, we can proceed as follows after determining $C_0$.
Let $[a_j^0, b_j^0]$ be the minimal interval so that
\[ \rho(\xi + d_{j-1}(s,t)) \geq C_0 - \alpha d_{j-1}(s,t) \quad \text{for all} \quad \xi \in (-\infty, a_j^0] \cup [b_j^0, \infty). \]
Define $[a_j^1, b_j^1]$ as per the previous definition \eqref{eq:interval_def}.
Then for $\delta \in (0,1)$, define $a_j^\delta = (1-\delta)a_j^0 + \delta a_j^1$, and similarly for $b_j^\delta$.
It is not hard to see that for any choice of $\delta$, using the intervals $[a_j^\delta, b_j^\delta]$ in the original algorithm provides a flow over time with cost horizon $C_0$; and moreover that the value of the flow depends continuously on $\delta$.
The proof of optimality is essentially unaffected.
So a further round of bisection or parametric search suffices here.

\medskip

Next, suppose we drop the requirement that $\rho$ is unimodal (maintaining still the growth assumption).
Instead, let us make the mild assumption that for any $z \in \Rplus$, $\rho^{-1}([0,z])$ is a union of finitely many compact intervals.
Little changes, except that instead of obtaining a single interval $[a_j, b_j]$ for each path $P_j$, we obtain a finite number of intervals.
\Cref{thm:fvalue} remains true as stated.

%% file: main.bbl
\begin{thebibliography}{10}

\bibitem{AMO93}
R.~K. Ahuja, T.~L. Magnanti, and J.~B. Orlin.
\newblock {\em Network Flows: Theory, Algorithms, and Applications}.
\newblock Prentice Hall, 1993.

\bibitem{ADL90}
R.~Arnott, A.~de~Palma, and R.~Lindsey.
\newblock Economics of a bottleneck.
\newblock {\em Journal of Urban Economics}, 27(1):111--130, 1990.

\bibitem{Baumann2006}
N.~Baumann and M.~Skutella.
\newblock Solving evacuation problems efficiently--earliest arrival flows with
  multiple sources.
\newblock In {\em Proceedings of the 47th Annual {IEEE} Symposium on
  Foundations of Computer Science {(FOCS)}}, pages 399--410, 2006.

\bibitem{BFA11}
U.~Bhaskar, L.~Fleischer, and E.~Anshelevich.
\newblock A {S}tackelberg strategy for routing flow over time.
\newblock In {\em Proceedings of the Twenty-Second Annual {ACM-SIAM} Symposium
  on Discrete Algorithms, {(SODA)}}, pages 192--201, 2011.

\bibitem{CCL11}
R.~Cominetti, J.~Correa, and O.~Larr{\'e}.
\newblock Dynamic equilibria in fluid queueing networks.
\newblock {\em Operations Research}, 63(1):21--34, 2015.

\bibitem{CCO17}
R.~Cominetti, J.~R. Correa, and N.~Olver.
\newblock Long term behavior of dynamic equilibria in fluid queuing networks.
\newblock In {\em Proceedings of the 19th International Conference on Integer
  Programming and Combinatorial Optimization (IPCO)}, pages 161--172, 2017.

\bibitem{Correa19}
J.~R. Correa, A.~Cristi, and T.~Oosterwijk.
\newblock On the price of anarchy for flows over time.
\newblock In {\em Proceedings of the 2019 {ACM} Conference on Economics and
  Computation, {(EC)}}, pages 559--577, 2019.

\bibitem{Disser}
Y.~Disser and M.~Skutella.
\newblock The simplex algorithm is {NP}-mighty.
\newblock In {\em Proceedings of the Twenty-Sixth Annual {ACM-SIAM} Symposium
  on Discrete Algorithms, {(SODA)}}, pages 858--872, 2015.

\bibitem{FleischerTardos98}
L.~Fleischer and E.~Tardos.
\newblock Efficient continuous-time dynamic network flow algorithms.
\newblock {\em Operations Research Letters}, 23(3):71--80, 1998.

\bibitem{FordFulkerson58}
L.~R. Ford and D.~R. Fulkerson.
\newblock Constructing maximal dynamic flows from static flows.
\newblock {\em Operations Research}, 6(3):419--433, 1958.

\bibitem{FordFulkerson62}
L.~R. Ford and D.~R. Fulkerson.
\newblock {\em Flows in Networks}.
\newblock Princeton University Press, 1962.

\bibitem{Gale}
D.~Gale.
\newblock Transient flows in networks.
\newblock {\em The Michigan Mathematical Journal}, 6(1):59--63, 1959.

\bibitem{Grinold}
R.~Grinold.
\newblock Infinite horizon programs.
\newblock {\em Management Science}, 18:157--170, 1971.

\bibitem{Harks}
T.~Harks.
\newblock Pricing in resource allocation games based on duality gaps.
\newblock Preprint, arXiv:1907.01976, 2019.

\bibitem{Jarvis}
J.~J. Jarvis and H.~D. Ratliff.
\newblock Some equivalent objectives for dynamic network flow problems.
\newblock {\em Management Science}, 28(1):106--109, 1982.

\bibitem{Koch2011}
R.~Koch and M.~Skutella.
\newblock Nash equilibria and the price of anarchy for flows over time.
\newblock {\em Theory of Computing Systems}, 49(1):71--97, 2011.

\bibitem{KRS2009}
E.~K{\"{o}}hler, R.~H. M{\"{o}}hring, and M.~Skutella.
\newblock Traffic networks and flows over time.
\newblock In {\em Algorithmics of Large and Complex Networks - Design,
  Analysis, and Simulation}, pages 166--196, 2009.

\bibitem{Megiddo}
N.~Megiddo.
\newblock Combinatorial optimization with rational objective functions.
\newblock In {\em Proceedings of the 10th Annual {ACM} Symposium on Theory of
  Computing, (STOC)}, pages 1--12, 1978.

\bibitem{Minieka73}
E.~Minieka.
\newblock Maximal, lexicographic, and dynamic network flows.
\newblock {\em Operations Research}, 21(2):517--527, 1973.

\bibitem{Roughgardenbook}
N.~Nisan, T.~Roughgarden, {\'{E}}.~Tardos, and V.~V. Vazirani, editors.
\newblock {\em Algorithmic Game Theory}.
\newblock Cambridge University Press, 2007.

\bibitem{RSB92}
H.~E. Romeijn, R.~L. Smith, and J.~C. Bean.
\newblock Duality in infinite dimensional linear programming.
\newblock {\em Mathematical Programming}, 53:79--97, 1992.

\bibitem{SeringKoch19}
L.~Sering and L.~V. Koch.
\newblock Nash flows over time with spillback.
\newblock In {\em Proceedings of the Thirtieth Annual {ACM-SIAM} Symposium on
  Discrete Algorithms, {(SODA)}}, pages 935--945, 2019.

\bibitem{SeringSkutella18}
L.~Sering and M.~Skutella.
\newblock Multi-source multi-sink {Nash} flows over time.
\newblock In {\em 18th Workshop on Algorithmic Approaches for Transportation
  Modelling, Optimization, and Systems, {(ATMOS)}}, pages 12:1--12:20, 2018.

\bibitem{Skutella2009}
M.~Skutella.
\newblock An introduction to network flows over time.
\newblock In {\em Research Trends in Combinatorial Optimization}, pages
  451--482, 2009.

\bibitem{Small2015}
K.~A. Small.
\newblock The bottleneck model: An assessment and interpretation.
\newblock {\em Economics of Transportation}, 4(1):110--117, 2015.

\bibitem{Vickrey69}
W.~Vickrey.
\newblock Congestion theory and transport investment.
\newblock {\em American Economic Review}, 59(2):251--60, 1969.

\bibitem{Wilkinson71}
W.~L. Wilkinson.
\newblock An algorithm for universal maximal dynamic flows in a network.
\newblock {\em Operations Research}, 19(7):1602--1612, 1971.

\bibitem{YangMeng98}
H.~Yang and Q.~Meng.
\newblock Departure time, route choice and congestion toll in a queuing network
  with elastic demand.
\newblock {\em Transportation Research Part B: Methodological}, 32(4):247--260,
  1998.

\bibitem{Zadeh}
N.~Zadeh.
\newblock A bad network problem for the simplex method and other minimum cost
  flow algorithms.
\newblock {\em Mathematical Programming}, 5:255--266, 1973.

\end{thebibliography}
